\documentclass[letterpaper,10 pt,conference]{ieeeconf}  

\usepackage{amsmath,amssymb,amsfonts}
\usepackage{graphicx} 
\usepackage{enumitem}
\usepackage{xcolor}
\usepackage{algorithm}
\usepackage[english]{babel}
\usepackage{amsthm}
\usepackage{bbold}
\usepackage{subcaption}

\newcommand{\nn}{{\nonumber}}

\newtheorem{theorem}{Theorem}
\newtheorem{remark}[theorem]{Remark}
\newtheorem{example}[theorem]{Example}
\newtheorem{lemma}[theorem]{Lemma}

\newtheorem{proposition}[theorem]{Proposition}

\newtheorem{assumption}{Assumption}

 \newcommand{\Rbb}{\mathbb R}

\newcommand{\xh}{\hat{x}}

\newcommand{\Au}{A_{UIO}}
\newcommand{\Buu}{B_{UIO}^u}
\newcommand{\Buy}{B_{UIO}^y}
\newcommand{\Du}{D_{UIO}}

\newcommand{\zero}{\mathbb{0}}

\DeclareMathOperator{\im}{Im}
\DeclareMathOperator{\rank}{rank}
\DeclareMathOperator{\Ker}{ker}

\usepackage{graphicx}
\graphicspath{ {./images/} }

 \title{\LARGE\bf  A data-driven approach to UIO-based fault diagnosis}

\author{Giulio Fattore  and Maria Elena Valcher
\thanks{G. Fattore and  M.E. Valcher are with the Dipartimento di Ingegneria dell'Informazione,
 Universit\`a di Padova,
    via Gradenigo 6B, 35131 Padova, Italy, e-mail:  \texttt{giulio.fattore@phd.unipd.it, meme@dei.unipd.it}}
   }
  \date{}

\begin{document}
\maketitle

\begin{abstract}
In this paper we propose a data-driven approach to the design of a residual generator, based on a dead-beat unknown-input observer, for linear time-invariant discrete-time state-space models, whose state equation is affected both by disturbances and by actuator faults.
We first review the model-based conditions for the existence of  such a residual generator, and then prove that under suitable assumptions on the collected historical data, we are both able to determine if the problem is solvable and to identify the matrices of a possible residual generator.
We propose an algorithm that, based only on the collected data (and not on the system description), is able to perform both tasks. An illustrating example and some remarks on limitations and possible extensions of the current results conclude the paper.
\end{abstract}

\section{Introduction}
Data-driven approaches to the solution of control problems are quite pervasive in nowadays literature. In particular,  data-driven fault detection (FD) methods have been the subject of a significant number of papers \cite{MDPI2022,DingJPC2014, DingAutomatica2014,DingBook,JoeQin2009,LundgrenJung2022} Since the early works of J.F. Dong and M. Verhaegen \cite{DongVerhaegen2009} and  of S.X. Ding et al. \cite{DingIFAC2011}, 
FD schemes that are designed directly based on collected data, by avoiding the preliminary step of system identification, have been proposed. 
Most of the available  results, however, focus on detection, rather than identification, \cite{DingIFAC2011,DingIFAC2011_2, DingAutomatica2014} and provide algorithms to design the matrices of a residual generator, based for instance on QR-decomposition and SVD \cite{DingAutomatica2014}. Also, necessary and sufficient conditions for the problem solvability have typically been expressed in terms of the original system matrices, and not on the available data. This means that the assumption that 
fault detection and identification (FDI) can be successfully performed on the system is taken for granted and not directly checked on data.

Leveraging  some recent results on data-driven unknown-input observer design (see \cite{Ferrari-Trecate} or \cite{TAC-UIOdisaro,DisaròValcherReducedUIO}), we propose a data-driven approach to the design of a residual generator, based on a dead-beat unknown input observer (UIO) \cite{Ansari2019,HaraDBO}, for a generic linear time invariant state-space model, whose state equation is affected both by disturbances and by faults.
We first
 provide   necessary and sufficient conditions for the problem solvability, by adopting a model-based approach that relies on classic results by J. Chen and R.J. Patton \cite{Chen}. Secondly, we show that under a suitable assumption on the  data, that ensures that they are representative of the original system trajectories, we derive 
necessary and sufficient conditions for the problem solvability in terms  of the available data. Such conditions are weaker than the conditions that guarantee the identifiability of the original system matrices.
Finally, we provide an algorithm to derive the matrices of a dead beat UIO-based residual generator that allows to identify any fault affecting the system.

As previously mentioned, the paper leverages  the results about full-order and reduced-order asymptotic observers proposed in  \cite{Ferrari-Trecate} or \cite{TAC-UIOdisaro,DisaròValcherReducedUIO}, but significantly advances them in three aspects: first, it provides a way to check on the data (only) the problem solvability. Secondly, it provides a practical algorithm to design the matrices of a possible dead-beat UIO (while  the previous references focused only on existence conditions). Thirdly, it extends the analysis to FDI.
The  more restrictive choice of focusing on dead-beat UIOs, rather than   asymptotic UIOs, is meant to provide  a cleaner set-up, that allows for exact solutions, since one does not need to account for the contribution of the estimation error when trying to identify the fault.
As discussed in Remark~\ref{rem:threshold} and in the Concluding remarks, the case of a residual generator based on an asymptotic UIO could be treated in a similar way, and would lead to very similar results.
The only increased complexity is related to the practical implementation  of   fault identification, as it would require some least squares estimation procedure,  to remove the effect of the estimation error. Other possible limitations of the current set-up that can be easily overcome are discussed in the Concluding remarks.

The paper is organised as follows. In Section~\ref{sec:PS} we introduce the overall set-up and formally state the FDI problem we address in this paper, that assumes that disturbances and faults affect only the state update.Section~\ref{sec:MB} recalls the model-based solution to the design of a residual generator based on  a dead-beat UIO, by suitably adjusting the one available in the literature for asymptotic UIOs. 
We also provide necessary and sufficient conditions for the existence of a residual generator, based on the dead-beat UIO, that allows to uniquely identify the fault from the residual signal. There are high chances that this specific result is already known in the literature, but since we were not able to find a reference, we provided a short proof to lay on  solid ground  the subsequent data-driven analysis that relies on this result.
Under a rather common assumption on the data (see Assumption~\ref{ass:PE}), that can be related to the persistence of excitation \cite{vanWaardeWillems,WillemsPE} of the system inputs, in Section~\ref{sec:DD} we first provide data-based necessary and sufficient conditions for the problem solvability, and then, by resorting to a couple of technical results, we provide a simple Algorithm that allows to first check on data the problem solvability conditions, and then provides matrices 
of a dead-beat UIO-based residual generator. 
An example and some final remarks conclude the paper.

{\bf Notation}.
$I_N$ is the identity matrix of dimension $N$, $\zero_N\in \Rbb^N$ and $\zero_{n \times m}\in \Rbb^{n \times m}$ are the $N$-dimensional vector and $n\times m$ matrix, respectively, with all entries equal to $0$.
In some occasions, when the size of the zero block may be easily deduced from the context, the suffix $n \times m$ is omitted.\\
Given a matrix $M\in \Rbb^{n \times m}$, we denote by $M^\top$ its transpose, and by $M^\dag$ its Moore-Penrose inverse \cite{BenIsraelGreville}. If $M$ is of full column rank (FCR), then
$M^\dag=[M^{\top} M]^{-1}M^{\top}.$ A similar expression can be provided in the case when $M$ is of full row rank (FRR).\\
Given an $n$-dimensional  vector sequence $\{\omega(k)\}_{k \in {\mathbb Z}_+}$ taking values in $ \Rbb^n$ and a positive integer $N$,  we set
\begin{equation} \label{eq:piledvector}
   \omega_N(k) \doteq\begin{bmatrix}
       \omega(k)\\
       \omega(k+1)\\
       \vdots\\
       \omega(k+N-1)
   \end{bmatrix}\in \Rbb^{Nn}.
\end{equation}
\smallskip

\section{UIO-Based Fault Detection and Isolation: Problem Statement} \label{sec:PS}
Consider a linear time-invariant (LTI) discrete-time dynamical system $\Sigma$, described by the state-space model:    
\begin{subequations} \label{eq:system}
\begin{align}
    x(k+1)&=Ax(k)+Bu(k)+Ed(k)+Bf(k),\\
        y(k)&=Cx(k),\qquad \qquad k\in {\mathbb Z}_+,
        \end{align}
\end{subequations}
where $x(k)\in \Rbb^n$, $u(k)\in\Rbb^m$, $y(k)\in \Rbb^p$, $d(k)\in \Rbb^r$ and $f(k)\in\Rbb^m$ are the state, input, output, disturbance and actuator fault signals, respectively. The   matrices involved in the system description have the following sizes: $A\in \Rbb^{n\times n}$, $B\in \Rbb^{n\times m}$, $C\in \Rbb^{p\times n}$ and $E\in \Rbb^{n\times r}$. 
It entails no loss of generality assuming that $E$ has rank $r$, since we can always reduce ourselves to this case (possibly by redefining the disturbance $d(t)$).\\

We assume for the  UIO-based residual generator $\hat{\Sigma}$ the following description   \cite{Chen,Frank}:
\begin{subequations} \label{eq:UIOresgen}
    \begin{align}
        z(k+1)&=\Au z(k)+\Buu u(k)+\Buy y(k), \label{eq:UIOresgenstate}\\
        \xh (k)&=z(k)+\Du y(k),\label{eq:UIOresgenout}\\
        r(k)&=y(k)-C\xh(k), \qquad \qquad k\in {\mathbb Z}_+, \label{eq:UIOresgenres}
    \end{align}
\end{subequations}
where $z(k)\in \Rbb^n$ is the residual generator state vector, $\xh(k)\in\Rbb^n$ denotes the state estimate and $r(k)\in \Rbb^p$ the residual signal. The real matrices $\Au, \Buu, \Buy$ and $\Du$ have dimensions compatible with the previously defined vectors.\\
If we let $e(k)\doteq x(k)-\xh(k)$ denote the state estimation error at time $k$, then our goals are:
\begin{itemize}
\item[i)] We want to ensure that there exists $k_0 >0$ such that, when the system is not affected by actuator faults, for every initial condition $e(0)$, and every input and disturbance sequences applied to the original system $\Sigma$, the estimate $e(k)$ is zero for every $k\ge k_0$ (and hence so is the residual signal). This means that 
equations \eqref{eq:UIOresgenstate} and \eqref{eq:UIOresgenout} describe a dead-beat UIO \cite{Ansari2019,HaraDBO} for $\Sigma$;
\item[ii)] If a fault affects the system from some $k_f\ge k_0$ onward, then the residual signal $r(k)$ becomes nonzero (fault detection) and the knowledge of the residual allows to uniquely identify the fault that has affected $\Sigma$ (fault identification).
\end{itemize}
\smallskip

\section{Model Based Approach}  \label{sec:MB}
\label{sec3}
\subsection{Dead-beat UIO} 
In the following we will steadily work under the following:
\begin{assumption} \label{ass:UIO} We assume that
\begin{enumerate}[label={\textbf{\Alph*.}}, ref={1.\Alph*}]
    \item $\rank(CE)=\rank(E)=r$; \label{ass:UIO1}
    \item $\rank\left(\begin{bmatrix}
        zI_n-A & -E \\
        C & \zero_{p\times r}
    \end{bmatrix}\right) = n+r$,  $\forall \ z \in {\mathbb C}, |z|\ne 0$.
    \label{ass:UIO3}
\end{enumerate}
\end{assumption}
Assumption~\ref{ass:UIO} corresponds to what we call the {\em strong* reconstructability property} of the triple $(A,E,C)$, that extends the well known {\em strong* detectability property} \cite{Darouach2,Hautus}.
\\
\begin{proposition}
\label{prop:DBUIO}
    There exists a dead-beat UIO for system $\Sigma$ described as in \eqref{eq:UIOresgenstate}$\div$\eqref{eq:UIOresgenout} if and only if Assumption~\ref{ass:UIO} holds.
\end{proposition}
\begin{proof}
    The proof is a slight modification of the one available for asymptotic UIOs (see \cite{Darouach2,Darouach}) and  we report it here in concise form.\\
By making use of the equations describing the system and the UIO, we deduce that the estimation error evolves according to the following difference equation:
\begin{eqnarray} 
        e(k+1) 
         &=&[(I_n-\Du C)A-\Au(I_n-\Du C) \nonumber \\
         &-&\Buy C] x(k)+ (I_n-\Du C)Ed(k) \nonumber\\
         &+& \Au e(k)+(I_n-\Du C)Bf(k) \nonumber \\
         &+& [(I_n-\Du C)B-\Buu]u(k).
\label{eq:errordyn}
\end{eqnarray}
If no fault affects the system $\Sigma$ (namely $f(k)=0$ for every $k\in {\mathbb Z}_+$), the estimation error converges to zero in a finite number of steps, independently of the initial conditions of the system and the UIO, of the input $u$ and of the disturbance $d$, if and only if 
 the matrices $A_{UIO}$, $B_{UIO}^u$, $B_{UIO}^y$ and $D_{UIO}$ satisfy the following constraints:
\begin{subequations}\label{eq:constraint}
\begin{align}
        &(I_n-\Du C)A-\Au(I_n-\Du C)=\Buy C \label{eq:constraint1}\\
        &\Buu=(I_n-\Du C)B \label{eq:constraint2}\\
        &(I_n-\Du C)E=\zero_{n \times r} \label{eq:constraint3}\\
        &\Au \ \text{nilpotent}
       \label{eq:constraint4}
\end{align}
\end{subequations}
By making use of the analysis in \cite{Darouach2,Darouach}, we can claim that conditions \eqref{eq:constraint1}$\div$\eqref{eq:constraint4} hold if and only if 
Assumption~\ref{ass:UIO} holds. 
\end{proof}

As a result of  Proposition~\ref{prop:DBUIO}, under  
Assumption~\ref{ass:UIO} 
there exist matrices $\Au$, $\Buu$, $\Buy$ and $\Du$ satisfying \eqref{eq:constraint1}$\div$\eqref{eq:constraint4} 
and hence the error dynamics becomes
$$e(k+1)=\Au e(k) + \Buu f(k),$$
where $\Au$ is nilpotent. In particular, if $f(k)=0$, the equation becomes:
$$e(k+1)=\Au e(k).$$
 This implies that after a finite number of steps the estimation error becomes a function of the fault signal only.

\begin{remark} Assumption~\ref{ass:UIO1} is 
necessary and sufficient for 
 the solvability of \eqref{eq:constraint3}.
 Once a solution to \eqref{eq:constraint3} is obtained, conditions \eqref{eq:constraint1} and \eqref{eq:constraint2} immediately follow.
The  general solution of  \eqref{eq:constraint3} is given in \cite{Darouach}. 
\\
If, in addition, 
  Assumption~\ref{ass:UIO3} holds, then  (see \cite{Darouach}, Theorem 2)   among the matrices $\Du$ satisfying \eqref{eq:constraint3} there exist some for which the pair $((I_n- \Du C)A, C)$ is reconstructable. This is always the case if  $\rank(I_n - \Du C)=n-r$ (which also means that $\Ker(I_n -\Du C) = \im(E)$). 
Lemma~\ref{lemmone} in the Appendix proves that
  if $\Du$ is a solution of \eqref{eq:constraint3} of rank $r$ (e.g., $\Du =E(CE)^\dag$), then   $\rank(I_n-\Du C)= n-r$. This implies that if we choose  a solution $\Du$ of rank $r$ for 
  \eqref{eq:constraint3}, then the pair $((I_n- \Du C)A, C)$ is reconstructable.\\
  When so,  a possible way to compute $\Au$ is the following one.  From \eqref{eq:constraint1}, we deduce that $\Au$ can be expressed as
\begin{equation}\label{eq:PA-LC}
    \begin{split}
        \Au
           &=(I_n-\Du C)A-(\Buy-\Au\Du)C\\
           &=(I_n-\Du C)A-LC,\nn
    \end{split}
\end{equation}
for  $L=\Buy-\Au\Du$. Consequently, we first choose $L$ so that $\Au = (I_n- \Du C)A - L C$ is nilpotent and then deduce $\Buy$ as
 $\Buy = L +\Au\Du$.
\end{remark}

To summarise, by making use of Theorem 2 in \cite{Darouach} and the assumption that $(A,E,C)$ is strong* reconstructable, we 
first find a solution 
$\Du$  of \eqref{eq:constraint3} of  rank $r$.
Then we determine $L$ such that $\Au = (I_n - \Du C) A - LC$ is nilpotent and finally we set   $\Buy = L +\Au \Du$ and   $\Buu = (I_n - \Du C)B$.
\\
In the sequel, we will always assume that $\rank (\Du)=r$.\\

\subsection{UIO-Based Fault Detection and Identification}
 
Under the hypothesis that $\Au$ is nilpotent, in the absence of faults, 
the estimation error $e(k)$ goes to zero in a finite number of steps $k_0$. So, if we assume that the time $k_f$ at which the fault affects the system is greater than or equal to $k_0$, then
 $e(k_f)=\zero_{n}$. Therefore, the dynamics of the estimation error for $k\ge k_f$ is
\begin{subequations}
    \begin{align}
    e(k+1)&=\Au e(k)+\Buu f(k)\label{eq:error fault}, \\
    r(k)&=Ce(k)\label{eq:residual},
    \end{align}
\end{subequations}
with zero initial condition $e(k_f)$, and for every $k>k_f$, the residual signal at time $k$ is 
\begin{equation}
    r(k)= \sum_{i= k_f}^{k-1} 
    C\Au^{k-i-1}\Buu f(i),\nn
\end{equation}
 in particular, 
    $r(k_f+1)=C\Buu f(k_f)$,
and hence, it is immediate to see that  we are able to uniquely identify every fault from the residual (with one step delay),
if and only if $C\Buu$ is FCR. It is worth  noting that
the full column rank condition on $C \Buu$ remains necessary and sufficient for the fault identification even if we   consider a longer time window. 

Indeed, if we assume as observation window $\{k_f+1,k_f+2,\dots, k_f+N\}$,
   then the family of residual signals $r(k), k\in \{k_f+1,k_f+2,\dots, k_f+N\}$, can be described as follows:
 \begin{equation}\label{eq:residualvector}
 r_N(k_f+1) = M_N f_N(k_f),\nn
\end{equation}
 where 
  \begin{equation*} 
 M_N \doteq\begin{small}\begin{bmatrix}
        C\Buu & & & \\
        C\Au\Buu&C\Buu& &\\
        \vdots&\ddots&\ddots&\\
        C\Au^{N-1}\Buu &\dots& C\Au\Buu &C\Buu
    \end{bmatrix}\end{small}
\end{equation*}
and the vectors $r_N(k_f+1)$ and $f_N(k_f)$ are defined as in \eqref{eq:piledvector}.\\
 We can identify  the vector $f_N(k_f)\in \Rbb^{Nm}$ from the residual vector $r_N(k_f+1)\in \Rbb^{Np}$ if and only if the matrix $M_N\in \Rbb^{Np \times Nm}$ is FCR and this is the case if and only if $C\Buu\in \Rbb^{p\times m}$ is FCR. 
A necessary and sufficient condition for the matrix $C\Buu$ to be FCR is given in Proposition~\ref{prop:CBuuFCR}, below.

\begin{proposition} \label{prop:CBuuFCR}
    Suppose that Assumption~\ref{ass:UIO1} holds.
    Then the following facts are equivalent:
    \begin{enumerate}[label={{\roman*)}}]
    \item\label{ass_1_prop_CBFCR} $\rank([\begin{smallmatrix} 
        CB&CE
    \end{smallmatrix}])=m+r$;
    \item \label{ass_2_prop_CBFCR}$C\Buu = C(I_n - \Du C) B$ is FCR.
    \end{enumerate}
\end{proposition}

\begin{proof}   
\ref{ass_1_prop_CBFCR} $\Rightarrow$ \ref{ass_2_prop_CBFCR}  follows immediately from Lemma~\ref{lemmone}.

   \ref{ass_2_prop_CBFCR}  $\Rightarrow$ \ref{ass_1_prop_CBFCR} \ Assume that $\rank([\begin{smallmatrix}
        CB&CE
    \end{smallmatrix}])$ is smaller than $m+r$. Then there exists a nonzero vector $v = [\begin{smallmatrix}
        v_B^\top & v_E^\top
    \end{smallmatrix} ]^\top$ such that
    \begin{equation}
        \begin{bmatrix}
        CB&CE
    \end{bmatrix} \begin{bmatrix} v_B\cr v_E\end{bmatrix} =\zero_p.
    \label{eq:vB1}
    \end{equation}
    Note that  $v_B$ cannot be zero, otherwise \eqref{eq:vB1} would become $CE v_E=\zero_p$ and this would contradict Assumption~\ref{ass:UIO1}.
    Consequently,
 \begin{eqnarray*}
 \zero_p &=& (I_p - C \Du) \begin{bmatrix}
        CB&CE
    \end{bmatrix} \begin{bmatrix} v_B\cr v_E\end{bmatrix} 
  =  C\Buu v_B.
 \end{eqnarray*}
    Since $v_B \ne \zero_m$, this implies that $C\Buu v_B$ is not FCR.
    \end{proof}
 \smallskip

The results so far obtained can be summarised in the following proposition.
\medskip

\begin{proposition}
    Under Assumption~\ref{ass:UIO}, there exists a residual generator, based on a dead-beat
    UIO and described as in \eqref{eq:UIOresgen}, for the discrete-time state-space model \eqref{eq:system}, that allows to uniquely identify an arbitrary fault affecting the actuators if and only if 
    $\rank\left([\begin{smallmatrix} CB& CE \end{smallmatrix}]\right)=m+r.$
\end{proposition}

It is worth noting that Assumption~\ref{ass:UIO1}
is encompassed in condition $\rank\left([\begin{smallmatrix} CB& CE \end{smallmatrix}]\right)=m+r.$ Therefore, an alternative way to state the previous result is the following one. In the sequel, we will steadily refer to this version of the problem solution in the model-based approach.

\begin{proposition}\label{prop:finalMB}
    The following conditions are equivalent:
    \begin{itemize}
        \item[i)] There exists a residual generator, based on a dead-beat
    UIO and described as in \eqref{eq:UIOresgen}, for the discrete-time state-space model \eqref{eq:system}, that allows to uniquely identify an arbitrary fault affecting the actuators;
    \item[ii)] Assumption~\ref{ass:UIO3} holds and 
    $\rank\left([\begin{smallmatrix} CB& CE \end{smallmatrix}]\right)=m+r.$
    \end{itemize}
\end{proposition}

\begin{remark} \label{rem:threshold}
 The system model we have considered introduces the simplifying assumption that the disturbance $d$ and the fault $f$ affect only the state update. Moreover, the disturbance effect on the state dynamics is weighted by the matrix $E$, while the actuator fault  is weighted by the state to input matrix $B$. 
 If, in addition,   unmodeled noise   affects the system dynamics or the residual generation,
 the FD  must be implemented through a threshold mechanism that is based on the estimated maximum norm that the residual component that depends on such noise can exhibit.
 In fact, we may assume that if the residual norm is below a certain threshold it is treated as zero, and the first time 
 $k= K^*$ the residual norm exceeds the threshold this means that at the previous time instant there was a fault.  When so, it makes sense to estimate the fault affecting the system based on the residuals collected on the finite window $[K^*, K^*+N-1]$ as
 $\hat f_N(K^*-1) = {\rm argmin}_f\| r_N(K^*) - M_N f\|_2$.
\end{remark}

\section{Data-Driven Approach} \label{sec:DD}
We assume to have collected offline input, state and output measurements from the system (affected by disturbances but not by faults) on a finite time window   of (sufficiently large) length $T$: $u_d \doteq \{u_d(k)\}_{k=0}^{T-2}$, $x_d \doteq\{x_d(k)\}_{k=0}^{T-1}$ and $y_d \doteq \{y_d(k)\}_{k=0}^{T-1}$.
Accordingly, we set
\begin{subequations}
\begin{align*}
    U_P&=\begin{bmatrix}
        u_d(0),\dots,u_d(T-2)
    \end{bmatrix}\in \Rbb^{m \times(T-1)}\\
      X_P&=\begin{bmatrix}
        x_d(0),\dots,x_d(T-2)
    \end{bmatrix}\in \Rbb^{n \times(T-1)}\\
        Y_P&=\begin{bmatrix}
        y_d(0),\dots,y_d(T-2)
    \end{bmatrix}\in \Rbb^{p \times(T-1)}\\
      X_f&=\begin{bmatrix}
        x_d(1),\dots,x_d(T-1)
    \end{bmatrix}\in \Rbb^{n \times(T-1)}\\
    Y_f&=\begin{bmatrix}
        y_d(1),\dots,y_d(T-1)
    \end{bmatrix}\in \Rbb^{p \times(T-1)}
    \end{align*}
\end{subequations}
Note that it not uncommon to assume that the state variable is accessible 
during the preliminary data collection process \cite{TAC-UIOdisaro,Ferrari-Trecate,DePersisTesi2020}. 
Moreover, if this were not the case,  state estimation based on data  could not  be possible, as  the   input and output data  
do not provide information about 
  the specific state realization $\Sigma$.

Also, even if we assume that  no  direct measurements of the disturbance sequence $d_d \doteq \{d_d(k)\}_{k=0}^{T-2}$ associated with the historical data is available, nonetheless, for the subsequent analysis, it is convenient to introduce the symbol
$$D_P=\begin{bmatrix}
        d_d(0),\dots,d_d(T-2)
    \end{bmatrix}\in \Rbb^{r \times(T-1)}.$$
In the sequel, we will make the following:
\begin{assumption}\label{ass:PE} The dimension $r$ of the disturbance vector $d$ is known\footnote{It is worthwhile remarking that $r$ can be estimated from data. Indeed, it coincides with the maximum value that the quantity $\rank\left(\left[\begin{smallmatrix} U_p
\cr  X_p\cr X_f\end{smallmatrix} \right]\right) - (n+m)$ can take,
based on  data collected through different experiments.}, and the $(m+r+n)\times (T-1)$ matrix 
    \begin{equation}
    \begin{bmatrix}
                U_p\\D_p\\ X_p
    \end{bmatrix}\ \text{is FRR}.\nn
    \end{equation}
    \end{assumption}
Note that this assumption holds true, in particular,  if the pair $(A, [B \ E ])$ is controllable and the sequence $\{u_d,d_d\}$ is persistently exciting of order $n+1$ (see \cite{vanWaardeWillems}, Theorem 1).
\smallskip

We want to prove that, under Assumption~\ref{ass:PE}, 
the two necessary and sufficient conditions 
for the existence of a residual generator, based on a dead-beat UIO, that we derived in Proposition~\ref{prop:finalMB} via a model-based approach find  immediate counterparts in terms of the data matrices $U_p,X_p, X_f, Y_p$ and $Y_f$.
This means that, under Assumption \ref{ass:PE}, the solvability conditions derived in a data-driven set-up are not more restrictive than those  one needs to verify when starting from the system matrices.  
\medskip

\begin{proposition}\label{prop:DDassumptions}
Under Assumption~\ref{ass:PE}, the following facts are equivalent:
\begin{enumerate}[label={{\roman*)}}]
\item \label{ass:1_prop_DDass}Assumption~\ref{ass:UIO3} holds and $\rank\left([\begin{smallmatrix} CB& CE \end{smallmatrix}]\right)=m+r$;
\item \label{ass:2_prop_DDass} The following conditions on the data matrices hold:
\begin{enumerate}[label={{ii\alph*)}}]
    \item\label{ass:2a_prop_DDass}
    \begin{align*}
  \rank\left(\begin{bmatrix}
        zX_p-X_f\\
        Y_p\\
        U_p
    \end{bmatrix}\right)=n+r+m,\ \forall z\in {\mathbb C}\setminus \{0\};\nn
\end{align*}
\item\label{ass:2b_prop_DDass}
\begin{align}
  \rank\left(\begin{bmatrix}
        X_p\\
        Y_f
    \end{bmatrix}\right)=n+r+m.\nn
\end{align}
\end{enumerate}
\end{enumerate}
\end{proposition}

\begin{proof} We first prove that Assumption~\ref{ass:UIO3} is equivalent to \ref{ass:2a_prop_DDass}. 
Since the data are derived from the state-space model 
\eqref{eq:system}, it follows that they satisfy the state equation and hence
\begin{align}\label{eq:data_BEA}
    X_f=AX_p+BU_p+ED_p.
\end{align}
Consequently,
\begin{align*}
    zX_p-X_f&=(zI_n-A)X_p-BU_p-ED_p
\\&= \begin{bmatrix}
        -B&-E&zI_n-A
    \end{bmatrix}\begin{bmatrix}
        U_p\\D_p\\X_p
    \end{bmatrix}.
\end{align*}
This implies
\begin{align}
    \begin{bmatrix}
        zX_p-X_f\\
        Y_p\\
        U_p
    \end{bmatrix}= \begin{bmatrix}
        -B&-E&zI_n-A\\
        \zero&\zero&C\\
        I_m&\zero&\zero
    \end{bmatrix}\begin{bmatrix}
        U_p\\D_p\\X_p
    \end{bmatrix}.\nn
\end{align}
By Assumption~\ref{ass:PE} the matrix of data on the right is of full row rank $m+r+n$, therefore condition \ref{ass:2a_prop_DDass} holds if and only if 
 $$\rank \left(\begin{bmatrix}
        -B&-E&zI_n-A\\
        \zero&\zero&C\\
        I_m&\zero&\zero
    \end{bmatrix}\right) = m+r+n,\ \forall z\neq0,$$
    which is equivalent to Assumption~\ref{ass:UIO3}.
    \smallskip
    \noindent We now show that $\left[\begin{smallmatrix} CB & CE \end{smallmatrix}\right]$ is FCR if and only if \ref{ass:2b_prop_DDass} holds.
    Making use, again, of the fact that the data are generated by the system \eqref{eq:system}, we can write:
   \begin{equation}
\begin{bmatrix}
    X_p\\Y_f
\end{bmatrix}= \begin{bmatrix}
     \zero_{n \times m} &\zero_{n \times r}&I_n\\
    CB &CE& CA
\end{bmatrix}\begin{bmatrix}
    U_p\\D_p\\X_p
\end{bmatrix}.\nn
\end{equation}
As the last matrix is FRR by Assumption~\ref{ass:PE},
condition \ref{ass:2b_prop_DDass} holds if and only if 
$$\rank \left(\begin{bmatrix}
     \zero_{n \times m} &\zero_{n \times r}&I_n\\
    CB &CE& CA
\end{bmatrix}\right) = m+r+n,$$
which in turn holds true if and only if 
$\left[\begin{smallmatrix} CB & CE \end{smallmatrix}\right]$ is FCR.
\end{proof}
\smallskip

The previous proposition is extremely useful because it allows one to immediately check from data if the problem is solvable. We now provide the path to explicitly determine a solution, provided that the aforementioned conditions hold.

\begin{proposition}\label{prop:eqconditions}
Under Assumption~\ref{ass:PE}, the following facts are equivalent:
 \begin{enumerate}[label={{\roman*)}}]
 \item \label{hp:1_prop_eqcond}
 Assumption~\ref{ass:UIO3} holds and $\rank\left([\begin{smallmatrix} CB& CE \end{smallmatrix}]\right)=m+r$;
\item \label{hp:2_prop_eqcond}
There exist real constant matrices $T_1$, $T_3$ and $T_4$, of suitable dimensions, such that $(T_3,C)$   is reconstructable, $\rank (CT_1) = m$, and
\begin{equation}\label{eq:DDsolution0}
    X_f=\begin{bmatrix}
      T_1&  T_3& T_4
    \end{bmatrix}\begin{bmatrix}
             U_p\\
          X_p\\
          Y_f
    \end{bmatrix};
\end{equation}
\item \label{hp:3_prop_eqcond}
There exist real constant matrices $T_1$, $T_2$, $T_3$ and $T_4$, of suitable dimensions, with $T_3$   nilpotent and $\rank (CT_1) = m$, such that
\begin{equation}\label{eq:DDsolution}
    X_f=\begin{bmatrix}
      T_1& T_2& T_3& T_4
    \end{bmatrix}\begin{bmatrix}
             U_p\\
          Y_p\\
          X_p\\
          Y_f
    \end{bmatrix}.
\end{equation}
\end{enumerate}
\end{proposition}

\begin{proof} We preliminarily observe that
  since  data are generated by the original system, we can claim that 
  \begin{equation}
\label{datalongvsshort}
\begin{bmatrix}
    U_p\\Y_p\\X_p\\Y_f
\end{bmatrix}= \begin{bmatrix}
     I_m &\zero_{m \times r}& \zero_{m \times n}\\
    \zero_{p \times m} &\zero_{p \times r}& C\\
    \zero_{n \times m} &\zero_{n \times r}&I_n\\
    CB &CE& CA
\end{bmatrix}\begin{bmatrix}
    U_p\\D_p\\X_p
\end{bmatrix},
\end{equation}
where the last matrix is of full row rank by Assumption~\ref{ass:PE}.
\smallskip

\noindent \ref{hp:1_prop_eqcond} $\Rightarrow$ \ref{hp:2_prop_eqcond} If $\rank \left(\left[\begin{smallmatrix} CB& CE \end{smallmatrix}\right]\right) = m+r$, then  $\rank (CE)=r$ and therefore not only Assumption~\ref{ass:UIO3} but also also Assumption~\ref{ass:UIO1} holds. Hence (see the previous section) there exists $\Du$ of rank $r$ such that $\Du CE =E$, 
and hence  $((I_n - \Du C)A,C)$ is reconstructable.
Since the data obey  \eqref{eq:data_BEA}, it follows that
$$(I - \Du C) X_f= (I_n - \Du C)A X_p + (I_n - \Du C) BU_p,$$
and hence \eqref{eq:DDsolution0} holds for 
 $${\small \begin{bmatrix}
        T_1 &T_3&T_4
    \end{bmatrix} = \begin{bmatrix}
            (I_n - \Du C)B & (I_n - \Du C)A & \Du
    \end{bmatrix}.}$$
    Clearly, $(T_3, C) = ((I_n - \Du C)A, C)$ is reconstructable. On the other hand, as $\rank \left(\left[\begin{smallmatrix} CB& CE \end{smallmatrix}\right]\right) = m+r$, by making use of Proposition 
\ref{prop:CBuuFCR} (\ref{ass_1_prop_CBFCR} $\Rightarrow$ \ref{ass_2_prop_CBFCR}) we can claim that $CT_1 = C(I_n - \Du C)B$ is FCR.
\smallskip

\noindent \ref{hp:2_prop_eqcond} $\Rightarrow$ \ref{hp:3_prop_eqcond}\ 
Suppose that $(\bar T_1, \bar T_3, \bar T_4)$ is a particular solution of \eqref{eq:DDsolution0} with $(\bar T_3, C)$ reconstructable and $C\bar T_1$ FCR,  and let $L$ be a matrix such that $\bar T_3 - LC$ is nilpotent. Then it is immediate to see that
$$
    X_f=\begin{bmatrix}
      \bar T_1& L& \bar T_3- LC & \bar T_4
    \end{bmatrix}\begin{bmatrix}
             U_p\\
          Y_p\\
          X_p\\
          Y_f
    \end{bmatrix}
$$ 
and hence \ref{hp:3_prop_eqcond} holds
for $(T_1,T_2,T_3,T_4)= (\bar T_1,L,\bar T_3 - LC,\bar T_4)$.
\smallskip

\noindent  \ref{hp:3_prop_eqcond} $\Rightarrow$ \ref{hp:1_prop_eqcond}\ 
If equation \eqref{eq:DDsolution} admits a solution $(T_1,T_2,T_3,T_4)$, then, by making use of  \eqref{datalongvsshort} and \eqref{eq:data_BEA} we can claim that
\begin{eqnarray*}
    X_f&=&
 {\small \begin{bmatrix}
      T_1& T_2& T_3& T_4
    \end{bmatrix}  \begin{bmatrix}
     I_m &\zero_{m \times r}& \zero_{m \times n}\\
    \zero_{p \times m} &\zero_{p \times r}& C\\
    \zero_{n \times m} &\zero_{n \times r}&I_n\\
    CB &CE& CA
\end{bmatrix} \begin{bmatrix}
    U_p\\D_p\\X_p
\end{bmatrix}} \\
&=& \begin{bmatrix} B & E & A \end{bmatrix} \begin{bmatrix}
    U_p\\D_p\\X_p
\end{bmatrix}.
\end{eqnarray*}
Since $ \left[\begin{smallmatrix}
    U_p\\D_p\\X_p
\end{smallmatrix}\right]$ is FRR, this implies that 
\begin{equation}
    \label{eq:BEA}
\begin{bmatrix}
      T_1& T_2& T_3& T_4
    \end{bmatrix}
    \begin{bmatrix}
     I_m &\zero& \zero\\
    \zero &\zero& C\\
    \zero &\zero&I_n\\
    CB &CE& CA
\end{bmatrix}
= \begin{bmatrix} B & E & A \end{bmatrix},
\end{equation}
which implies, in particular, that $T_4 CE =E$ (or equivalently $(I_n-T_4C)E =\zero_{n\times r}$) which proves \ref{ass:UIO1}.
\\ On the other hand, the nilpotent matrix $T_3$ can be expressed as $T_3= (I_n-T_4C)A - T_2 C.$
To prove that Assumption~\ref{ass:UIO3} holds, observe that
\begin{eqnarray*}
    n \!\!\!&+ &\!\!\! r \ge  \rank \left(\begin{bmatrix} z I_n -A & - E\cr C &\zero_{p \times r}\end{bmatrix}\right) \\
    &\ge&\!\!\! 
    \rank \left(\begin{bmatrix}
        I_n - T_4 C & \zero_{n \times p}
        \cr
        E^\dag & \zero_{r \times p}\cr \zero_{p \times n} &I_p
    \end{bmatrix}\begin{bmatrix} z I_n -A & - E\cr C &\zero_{p\times r}\end{bmatrix}\right)\\
  &=& \!\!\!
    \rank \left(\begin{bmatrix}
        (I_n - T_4 C)z - T_2C - T_3 & \zero_{n \times r}
        \cr
        * & - I_r\cr C &\zero_{p \times r}\end{bmatrix}\right)  \\
     &=& \!\!\!
   {\small \rank \left(
    \begin{bmatrix} I_n &\zero& - T_2 -T_4 z\cr \zero & I_r &\zero
    \cr \zero & \zero &I_p
\end{bmatrix}
    \begin{bmatrix}
        z I_n -  T_3 & \zero
        \cr
        * & - I_r\cr C &\zero\end{bmatrix}\right)}    \\
        &=& \!\!\!
    \rank \left(
    \begin{bmatrix}
        z I_n -  T_3 & \zero_{n \times r}
        \cr
        * & - I_r\cr C &\zero_{p \times r}\end{bmatrix}\right) = n+r, \forall \ z\ne 0.
\end{eqnarray*}
Finally, we want to prove that $\rank (CT_1)=m$ implies $\rank \left(\left[\begin{smallmatrix} CB& CE \end{smallmatrix}\right]\right) = m+r$.
From \eqref{eq:BEA} one easily gets 
\begin{equation*}
\begin{bmatrix}
      CT_1& C T_4
    \end{bmatrix}\begin{bmatrix}
     I_m &\zero_{m \times r}\\
    CB &CE
\end{bmatrix} = \begin{bmatrix} CB & CE  \end{bmatrix},
\end{equation*}
and hence
\begin{equation}
\begin{bmatrix}
      CT_1& \zero_{p\times r}
    \end{bmatrix} = (I_p - C T_4) \begin{bmatrix}
    CB &CE
\end{bmatrix}.\nn
\end{equation}
The rest of the proof is analogous to the one of \ref{ass_2_prop_CBFCR}  $\Rightarrow$ \ref{ass_1_prop_CBFCR} in Proposition~\ref{prop:CBuuFCR}.

\end{proof}
\medskip

It is worthwhile at this point to summarize the outcome of the previous analysis. In Section~\ref{sec3} we have proved (see Proposition~\ref{prop:finalMB}) that a residual generator, based on a dead-beat
UIO and described as in \eqref{eq:UIOresgen}, exists if and only if Assumption~\ref{ass:UIO3} and $\rank \left(\left[\begin{smallmatrix} CB& CE \end{smallmatrix}\right]\right) = m+r$ hold. In Proposition~\ref{prop:DDassumptions}
we showed how such conditions can be equivalently checked on the available data. Finally, Proposition~\ref{prop:eqconditions} tells us that such conditions are equivalent to the existence of a solution $(T_1, T_2, T_3,T_4)$
of \eqref{eq:DDsolution}, with $T_3$   nilpotent and $\rank (CT_1) = m$.
By resorting to the results derived in
 \cite{TAC-UIOdisaro} (see, also, \cite{Ferrari-Trecate}), 
we can immediately determine
the matrices  $\Au$, $\Buu$, $\Buy$ and $\Du$ of the (dead-beat) UIO \eqref{eq:UIOresgenstate}$\div$\eqref{eq:UIOresgenout}  as
\begin{equation}\label{eq:UIOmat_DD}
\begin{array}{rl}
\Au= T_3, \quad &  \Du = T_4, \cr 
\Buu = T_1, \quad& \Buy = T_2 + T_3T_4.
\end{array}
\end{equation}
Moreover, as a consequence of Assumption~\ref{ass:PE}, we can claim that $X_p$ is FRR, and hence by exploiting the relationship $Y_p = C X_p$ we can uniquely identify $C$ as $C = Y_p X_p^\dag$. This allows us to also determine the expression of the residual \eqref{eq:UIOresgenres}.

Proposition~\ref{prop:eqconditions} together with identities \eqref{eq:UIOmat_DD} allow to derive the matrices of the 
desired residual generator. However, 
once the necessary and sufficient conditions for the problem solvability have been checked on the historical data, 
it is not obvious how to identify, among all the solutions 
$(T_1, T_2,T_3, T_4)$ of 
\eqref{eq:DDsolution}, one with
 $T_3$   nilpotent and $CT_1$ FCR.
To this end we want to provide an algorithm which is based on a first important observation, condensed in the following lemma.
\medskip

\begin{lemma} \label{lemma10}
The following facts are equivalent:
\begin{enumerate}[label={{\roman*)}}]
\item \label{hp:1_lemma10}
There exists a solution $(T_1, T_3, T_4)$ of \eqref{eq:DDsolution0},  for which $(T_3,C)$   is reconstructable and $\rank(CT_1)=m$;
\item \label{hp:2_lemma10} There exists a solution $(T_1, T_3, T_4)$  of \eqref{eq:DDsolution0}, for which   $\rank (T_4) = r$, and for every such solution   the pair $(T_3,C)$   is reconstructable and $\rank(CT_1)=m$.
\end{enumerate}
\end{lemma}

\begin{proof} 
\ref{hp:2_lemma10} $\Rightarrow$ \ref{hp:1_lemma10} is obvious, so we only need to prove that \ref{hp:1_lemma10} $\Rightarrow$ \ref{hp:2_lemma10}.
By mimicking the proof of \ref{hp:3_prop_eqcond} $\Rightarrow$ \ref{hp:1_prop_eqcond} in Proposition~\ref{prop:eqconditions}, we 
can claim  that a triple $(T_1,T_3,T_4)$ solves \eqref{eq:DDsolution0}  if and only if it solves
\begin{equation*}
\begin{bmatrix}
      T_1& T_3& T_4
    \end{bmatrix}\begin{bmatrix}
     I_m &\zero_{m \times r}& \zero_{m \times n}\\
    \zero_{n \times m} &\zero_{n \times r}&I_n\\
    CB &CE& CA
\end{bmatrix} = \begin{bmatrix} B & E & A \end{bmatrix},
\end{equation*}
namely if and only if 
$(T_1,T_3,T_4) = ((I_n - T_4C)B , (I_n - T_4C) A, T_4)$, with
$T_4$ any matrix satisfying $E = T_4 CE$.\\

Now, let $(\bar T_1,\bar T_3,\bar T_4)$  be any solution of \eqref{eq:DDsolution0}, satisfying the hypotheses in \ref{hp:1_lemma10}.
This implies that 
$\bar T_4CE =E$ (and hence $\rank(CE)=\rank (E)=r$),
the  pair $((I_n - \bar T_4C)  A, C) = (\bar T_3, C)$ is reconstructable, and 
$C(I_n - \bar T_4C)B = C \bar T_1$ has FCR.
But by the analysis we carried out in Section~\ref{sec3}, we can claim 
that there exists $T_4^*$ of rank $r$, satisfying $T_4^*CE =E$, and for every such $T_4^*$ it will still be true that  $((I_n - T_4^* C)  A, C)$ is reconstructable, and 
$C(I_n -  T_4^*C)B$ has FCR.
This completes the proof.
\end{proof} 
\medskip

The previous Lemma~\ref{lemma10} tells us that in order to find a solution $(T_1,T_3, T_4)$ of \eqref{eq:DDsolution0} with the required properties we need to focus on those for which  $\rank (T_4) = r$. 
On the other hand, the previous proof has clarified that the solutions of \eqref{eq:DDsolution0}
are those and only those 
expressed as $(T_1,T_3,T_4) = ((I_n - T_4C)B , (I_n - T_4C) A, T_4)$, with
$T_4$ any matrix satisfying $E = T_4 CE$.
However, the matrix $E$ is not available, and hence we need to select such matrices $T_4$ by making use only of data.
\\
Let $S$ be a $(T-1)\times (T-1)$ nonsingular square (NSS) matrix such that 
\begin{equation}
\label{eq:algstep2}
    \begin{bmatrix}
       U_p\\ X_p\\Y_f
    \end{bmatrix} S = \begin{bmatrix}
       I_m & \zero_{m \times (T-1-m-n)} & \zero_{m \times n} \cr
       \zero_{n \times m} & \zero_{n \times (T-1-m-n)} & I_n \\
       Y_B & Y_E & Y_A
    \end{bmatrix},
\end{equation}
for suitable matrices $Y_A, Y_B$ and $Y_E$ with $p$ rows.
Such a matrix exists because $[\begin{smallmatrix}
       U_p\\ X_p
    \end{smallmatrix}]$ is FRR (as a consequence of Assumption~\ref{ass:PE}).
Set $\bar X_f \doteq X_f S$ and block-partition it, according  to the block-partitioning of $Y_f S$, as
$$\bar X_f = \begin{bmatrix} X_B & X_E & X_A\end{bmatrix}.$$ Then, $(T_1,T_3,T_4)$ solves \eqref{eq:DDsolution0}  if and only if it solves 
\begin{align*}
&\begin{bmatrix} X_B & X_E & X_A\end{bmatrix} = \\
&=\begin{bmatrix} T_1 & T_3 & T_4\end{bmatrix} \begin{bmatrix}
       I_m & \zero_{m \times (T-1-m-n)} & \zero_{m \times n} \cr
       \zero_{n \times m} & \zero_{n \times (T-1-m-n)} & I_n \\
       Y_B & Y_E & Y_A
    \end{bmatrix},
\end{align*}
namely if and only if 
$(T_1,T_3,T_4) = (X_B - T_4 Y_B, X_A - T_4 Y_A, T_4)$, with
$T_4$ any matrix satisfying $X_E = T_4 Y_E$.\\
This allows to say that 
$T_4$  satisfies $X_E = T_4 Y_E$ if and only if it satisfies
$E = T_4 CE$.

The previous analysis and results lead to Algorithm~\ref{alg:DDUIO} that describes  the procedure to determine the matrices of a dead-beat UIO-based residual generator ($\Au$, $\Buu$, $\Buy$, $\Du$, $C$) for system \eqref{eq:system} (under Assumption~\ref{ass:PE}), based on the given data.
 
\begin{algorithm}
    \caption{Data Driven UIO matrix estimation procedure}\label{alg:DDUIO}
     \textbf{Input}: $r$, $U_p$, $X_p$, $Y_p$, $X_f$, $Y_f$;\\
  \textbf{Output}: $\Au$, $\Buu$, $\Buy$, $\Du$, $C$;
  \begin{enumerate}
 \item[\textbf{1.}] Check if conditions \ref{ass:2a_prop_DDass} and \ref{ass:2b_prop_DDass} in Proposition~\ref{prop:DDassumptions} hold.
If not,   the problem is not solvable. Otherwise, go to Step~2.
\smallskip
\item[\textbf{2.}] Set $C = Y_p X_p^\dag$. Let $S$ be a $(T-1)\times (T-1)$ NSS matrix such that \eqref{eq:algstep2} holds, for suitable matrices $Y_A, Y_B$ and $Y_E$ with $p$ rows.
Find a solution $T_4^*$ of $X_E = T_4 Y_E$ with 
  $\rank (T_4^*) = r$, and set $T_1^* = X_B - T_4^* Y_B$ and $T_3^* = X_A - T_4^* Y_A$.
Since at Step 1 we have checked that the problem is solvable, then necessarily $(T_3^*, C)= (X_A - T_4^* Y_A, C)$ is  reconstructable and $CT_1^*$ is FCR.\\
\smallskip
\item[\textbf{3.}]  Let $L$ be such that $T_3^*  - L C$ is nilpotent. Then the matrices of the residual generator are:
$$
\begin{array}{l}
\Au= T_3^* - LC, \quad    \Du = T_4^*,  \quad C = Y_p X_P^\dag, \cr\cr
\Buu = X_B - T_4^* Y_B, \quad \Buy = L + (T_3^*-LC) T_4^*.
\end{array}$$
    \end{enumerate}
\end{algorithm}
\smallskip

\begin{example} \label{ex1}
Consider the system $\Sigma$, of dimension $n=5$,
with describing matrices:

\begin{small}
\begin{align*}
A  &= \begin{bmatrix}
    0.8&  0&    0&      0&    0\\
   -0.8  &0 &   0&      0&    0\\
   -1   & 0 &  -1.2  & -0.5 &-1.3\\
    2   &-0.6 & 2.6 &   1 &   2.3\\
    0.8 &-0.9&  0.6&    0.1&  0
\end{bmatrix},\\
B &= \begin{bmatrix}
    1\\
    0\\
    0\\
    0\\
    0
\end{bmatrix}, \  E= \begin{bmatrix}
    0 &0\\
    0 &0\\
    1 &1\\
    0 &1\\
    0& 0
\end{bmatrix}, \
C = \begin{bmatrix}
    1 &0& 0&  0& 0\\
   0 &0 &1 &-2& 0\\
  -1 &0 &0 & 1& 0
\end{bmatrix}.
\end{align*}
\end{small}

\noindent As   in \cite{DisaròValcherReducedUIO}, the historical  (both known and unknown) input data have been randomly generated, uniformly in the interval $(-5,5)$ for the known input $u(k)$, and in the interval $(-2,2)$ for the disturbance $d(k)$. The time-interval of the offline experiment has been set to $T=150$. We have collected the data corresponding to the input/output/state trajectories and then checked that assumptions \ref{ass:2a_prop_DDass} and \ref{ass:2b_prop_DDass}.
 Clearly, from $Y_p$ and $X_p$ we   deduce  $C$.
We have then set as matrices of the UIO in \eqref{eq:UIOresgen} the ones corresponding to the following particular solution of equation \eqref{eq:constraint}, namely
\begin{small}
\begin{align*}
\Au &= \begin{bmatrix}
    0  & 0 &   0&  0& 0\\
   -0.8& 0 &   0 & 0& 0\\
    0.8& 0 &   0 & 0& 0\\
    -1.6& 0 &   0 & 0& 0\\
    0.8& -0.9& 0.6 &0.1 & 0\\
\end{bmatrix}, \ \Buu = \begin{bmatrix}
   1\\
   0\\
   1\\
   -2\\
   0\\
\end{bmatrix},\\
\Buy &= \begin{bmatrix}
       0.8 & 0 & 0\\
         0 & 0 & 0\\
         0 & 0 & 0\\
         0 & 0 & 0\\
         0.9 & 0.6 &  1.3\\
\end{bmatrix}, \ \Du = \begin{bmatrix}
     0 & 0 & 0\\
     0 & 0 & 0\\
     1 & 1 & 2\\
     3 & 0 & 1\\
     0 & 0 & 0\\
\end{bmatrix}.  
\end{align*}
\end{small}

It is easy to verify that the matrix $\Au$ is nilpotent with nilpotency index $k_N=3$. \\
In order to test the performance of the dead-beat UIO generator, we fed the system with the (known) input $u(k) = 0.9\sin(0.4k+3), \ k\in\mathbb{Z}_+$, a random disturbance $d(k)$ whose first and second entries take values uniformly in the interval $(-5,5)$ and $(-2,2)$, respectively, and the fault signal $$f(k)=\begin{cases}
    0&k<k_f,\\
    \max\left\{0.1+\exp\left(\frac{-10}{k-k_f+1}\right),0.9\right\}& k\ge k_f.
\end{cases}$$ 
We let $k_{id}$ denote the first time  at which we start the fault estimation. We have simulated four different scenarios in which we changed the relative position of $k_N,k_{id}$ and $k_f$. 

\begin{figure}
\centering
\subfloat[$3=k_{id}\leq k_f=10$]{
  \includegraphics[width=40mm]{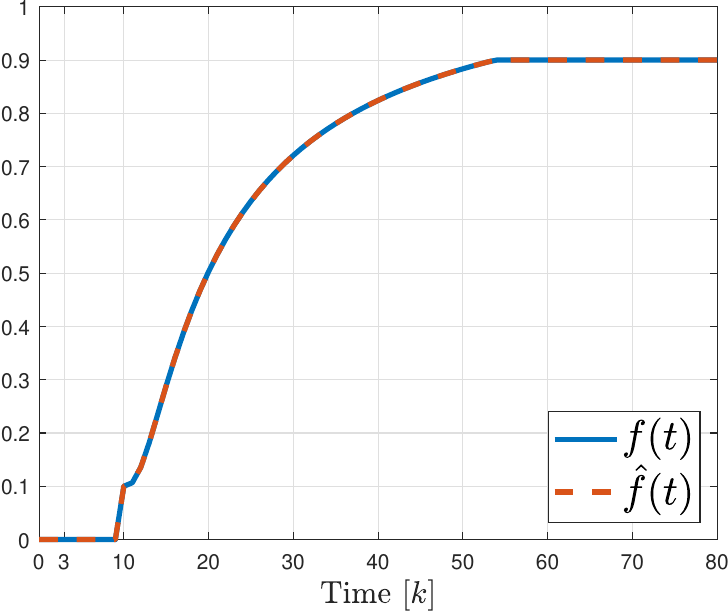}
}
\subfloat[$1=k_f< k_{id}=3$]{
  \includegraphics[width=40mm]{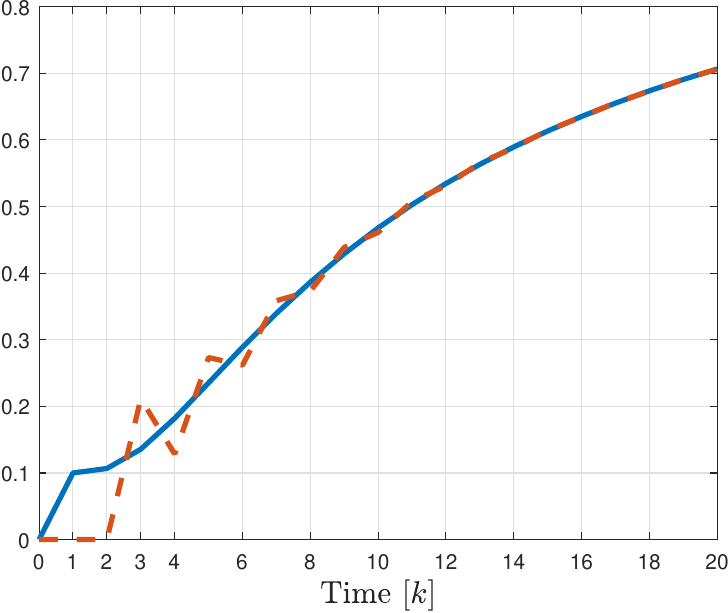}
}
\vspace{3mm}
\subfloat[$1=k_{id}< k_0=3$]{
  \includegraphics[width=40mm]{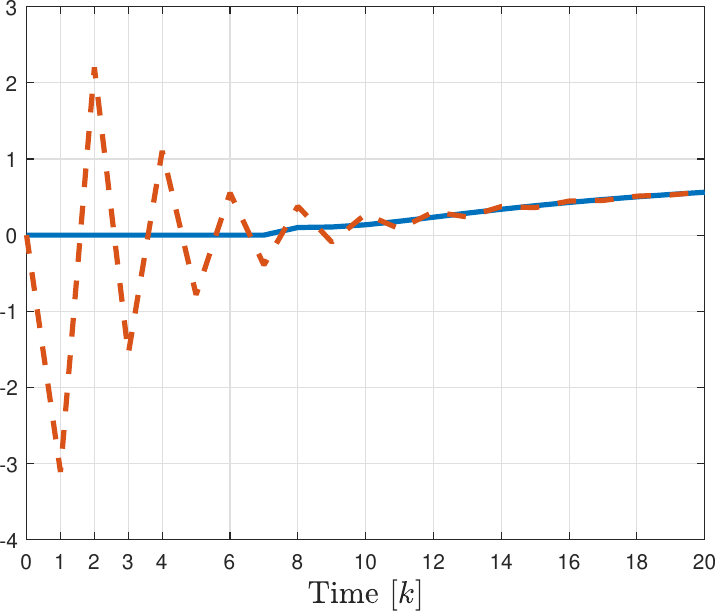}
}
\subfloat[$ 5=k_f<k_{id}=20$]{
  \includegraphics[width=40mm]{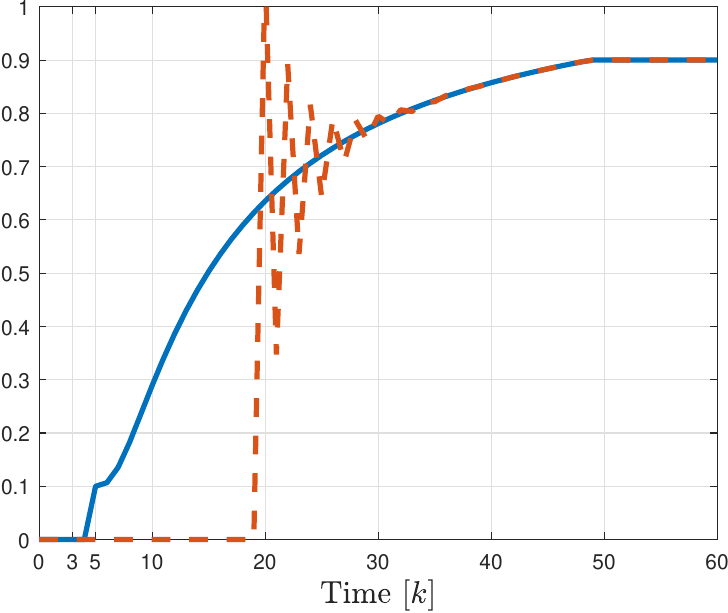}
}
\caption{Dynamics of the real and estimated faults.}
\label{fig1}
\end{figure}
On the top left plot we have the ideal case, namely, the one in which the identification starts only when the estimation error has become zero ($k_{id}=k_N$) and the fault   starts after $k_{id}$. In this scenario, the fault is perfectly estimated ($\hat f(k)=f(k), \forall k\ge k_{id}$). On the top right graph, we have simulated the case in which the fault appears before  $e(k)$ becomes zero ($k_f<k_N$), and the identification starts at $k_{id}=k_N$.  In the bottom left plot, we start the identification before $k_N$, when    the residual is nonzero only due to the estimation error ($k_{id}<k_N$), and the fault affects the system at $k_f> k_N$.  Finally, in the bottom right plot, the identification starts after  the fault appears and the effect of $e(0)$ is extinguished ($k_N\leq k_f<k_{id}$). In the last three cases we can notice that the fault estimate $\hat f$ oscillates around the real value of  $f$ and, in a finite number of steps, it reaches the correct value, as expected. \\
 From the third plot we deduce that starting the FDI before $k_N$ may lead to false alarms; on the other hand, from the fourth plot we can see that late estimation 
 leads to errors.
 Therefore 
the best choice is to start the identification   exactly when  $e(k)$   becomes zero, namely to assume $k_{id}=k_N$.   
\end{example}

\section{Concluding remarks}  \label{sec:CR}

In this paper we have addressed the data-driven FDI problem for a discrete-time LTI system affected by disturbances. The FD is performed through  a dead-beat UIO-based residual generator that is designed by resorting only to data, with no a priori information about the system matrices.
It is worth mentioning that also the conditions for the existence of such a residual generator can be checked on data.
Moreover, the available data allow to design the residual generator, but not to identify the original system description. An algorithm illustrates  the preliminary verification and the design procedure.

The paper introduces strong assumptions on the system in order to be able to identify (and not only detect) each possible fault affecting the actuators by resorting to a single residual generator. This seemed to us the natural starting point to illustrate our approach. However, the approach and the analysis we proposed can be generalised (and this is an ongoing research) in various ways:
first of all, we can resort to a bank of residual generators rather than a single one in order to identify the fault. Secondly, we may assume that the fault affects a single actuator at a time, since the case when multiple actuators fail at the same time is not always realistic. Finally, 
we have focused on the case when the UIO is dead-beat, since in this case the effect of the state estimation error on the residual vanishes in a finite number of steps. If we deal with an asymptotic UIO, we simply need to introduce some threshold mechanism and some least squares approximation, along the lines illustrated in Remark~\ref{rem:threshold}.
All such extensions are feasible at limited  effort in terms of analytical tools, and they would impose weaker conditions on the system properties and on the collected data, but of course they would make the 
formulas and calculations more complicated.

\appendix

\begin{lemma} \label{lemmone}
Let $A\in {\mathbb R}^{n\times p}, B\in {\mathbb R}^{p\times n},$ and $C\in {\mathbb R}^{n\times r},$  and assume that the following conditions hold:
\begin{enumerate}[label={c1)}]
\item \label{hp:c1} $\rank(C)=r$;
\end{enumerate}
\begin{enumerate}[label={c2)}]
\item \label{hp:c2} $(I_n-AB)C=\zero_{n\times r}$;
\end{enumerate}
\begin{enumerate}[label={c3)}]
\item \label{hp:c3} Either $\rank(A) =r$ or $\rank( B)=r$.
\end{enumerate}
Then
\begin{enumerate}[label={{\roman*)}}]
\item\label{hp:1_lemmone}$\rank(I_n-AB)=n-r$ and therefore $\im (C)= \Ker (I_n - AB)$.
\end{enumerate}
If, in addition, $D\in {\mathbb R}^{n\times m}$ is such that
\begin{enumerate}[label={c4)}]
\item\label{hp:c4} $\rank(\left[\begin{smallmatrix} C & D\end{smallmatrix} \right])=r + m$
\end{enumerate}
then
\begin{enumerate}[label={ii)}]
\item\label{hp:2_lemmone} $\rank((I_n-AB)D) =m$.
\end{enumerate}
\end{lemma}
\begin{proof} \ref{hp:1_lemmone}\ 
Conditions \ref{hp:c1} and \ref{hp:c2} imply that
$\rank(I_n-AB) \le n-r$.
From \ref{hp:c3} it follows that $k := \rank(AB) \le r$ and hence there exists a NSS matrix $T\in {\mathbb R}^{n\times n}$ such that
$$T^{-1} AB T = \begin{bmatrix}
    \Delta_1 & \zero_{k\times n-k}\cr
    \Delta_2 & \zero_{n- k\times n-k}
\end{bmatrix}, $$
where $\Delta_1\in {\mathbb R}^{k\times k}$ and 
$\Delta_2\in {\mathbb R}^{n-k\times k}$.
This implies that
\begin{eqnarray*}
    \rank (I_n-AB)&=& \rank(T^{-1} (I_n-AB) T) \\
    &=& \rank \left(\begin{bmatrix} I_k - \Delta_1 & 
    \zero_{k\times n-k}\cr
    - \Delta_2 & I_{n- k}
    \end{bmatrix}\right)
\end{eqnarray*}
This implies that $\rank (I_n-AB) \ge n-k \ge n-r$. 
Consequently, $\rank (I_n-AB) = n-r$ and the second claim follows immediately from \ref{hp:c2}. 
\smallskip

\ref{hp:2_lemmone}\ Let
$$ S= \begin{bmatrix} P\cr R\end{bmatrix} \in {\mathbb R}^{n\times n},$$
with $P\in {\mathbb R}^{r \times n}$ and $R\in {\mathbb R}^{(n-r)\times n}$, be a NSS  matrix such that
$$S C = \begin{bmatrix} I_r\cr \zero_{n-r\times r}\end{bmatrix}.$$
Then, by \ref{hp:c4}, 
$$r+m = \rank \left(S\begin{bmatrix}C & D\end{bmatrix}\right) = \rank\left(
\begin{bmatrix}
  I_r & PD\cr
    \zero_{n- r \times r} & RD
\end{bmatrix}\right)$$
implies $\rank(RD)=m$.
Moreover, as $\im (C)= \Ker (I_n - AB)$, it follows that
$I_n - AB = L R$
for some FCR matrix $L\in {\mathbb R}^{n\times (n-r)}.$
Finally, 
$(I_n - AB) D = L (RD)$ is the product of two FCR matrices, and hence has FCR $m$.
\end{proof}

\end{document}